\documentclass[12pt]{article}

\usepackage{amsmath,amsthm,amsfonts,amssymb,amscd,enumerate,graphicx,color}

\allowdisplaybreaks[4]
\newtheorem {lemma}{Lemma}[section]
\newtheorem {theorem} {Theorem}[section]
\newtheorem{proposition}{Proposition}[section]

\numberwithin{equation}{section}

\usepackage[T1]{fontenc}
\usepackage[utf8]{inputenc}
\usepackage{authblk}

\allowdisplaybreaks [4]
\begin{document}

\title{Minimum status, matching and domination of graphs}

\author[a]{Caixia Liang} 
\author[b]{Bo Zhou\footnote{Corresponding author. E-mail: zhoubo@scnu.edu.cn}}
\author[b]{Haiyan Guo}
\affil[a]{School of Mathematics and Statistics,
Zhaoqing University, Zhaoqing 526061, P.R. China}
\affil[b]{School of Mathematical Sciences,
South China Normal University, Guangzhou 510631, P.R. China}



\date{}
\maketitle

\begin{abstract}
Given a graph, the status of a vertex  is the sum of the distances between the vertex and all other vertices. The minimum status of a graph  is  the minimum of statuses of all vertices of this graph. We give a sharp upper bound for the minimum status of a connected graph with fixed  order and matching number (domination number,  respectively), and characterize the unique trees achieving the bound. We also determine the unique tree such that its  minimum status is as small as possible when   order and matching number (domination number,  respectively) are fixed.
\end{abstract}

\section{Introduction}

All graphs considered in this paper are finite,  simple and connected.  Let $G$ be a  graph with vertex set $V(G)$ and edge set $E(G)$.  For $u\in V(G)$,  the status of $u$ in $G$  is defined as the sum of the distances from $u$ to all other vertices in $G$,  denoted by $s_G(u)$ or $s(u)$,  see \cite{BH}.  That is, $s_G(u)=\sum_{v\in V(G)}d_G(u, v)$,  where $d_G(u, v)$ is the distance between vertices $u$ and $v$ in $G$,  i.e.,   the length (number of edges)  of a shortest path connecting $u$ and $v$ in $G$.  The status of a vertex is also known as its transmission \cite{Do, KS} or its total distance \cite{BBC}.  The  minimum status of $G$,  denoted by  $s(G)$,   is defined as
\[
s(G)=\min\{s_G(u): u\in V(G)\}.
\]
The study of this concept is essentially the study of the the proximity: If the order of $G$ is $n$,  then the proximity of $G$ is defined as $\pi (G)=\frac{1}{n-1}s(G)$ \cite{A,  ACH,  AH}.  The set $\{u\in V(G): s_G(u)=s(G)\}$ is called the median of $G$.

For a graph $G$ with $u \in V(G)$,   $G-u$ denotes the graph obtained from $G$ by deleting $u$ (and its incident edges).  Let $T$ be a tree.  For $u\in V(T)$,  the components of $T-u$ are known as the  branches of $T$ at $u$.  The branch-weight of a vertex $u$ in $T$  is the maximum number of vertices of branches of $T$ at $u$,  denoted by $w_T(u)$.  The  minimum branch-weight of $T$ is defined as
\[
w(T)=\min\{w_T(u): u\in V(T)\}.
\]
The set $\{u\in V(T):  w_T(u)=w(T)\}$ is called the centroid of $T$.  Zelinka \cite{Zelinka} showed that in a tree a vertex  is in  its median if and only if it is in its centroid,  i.e.,   for any tree,  its median is equal to its  centroid,  see also \cite{KA}.

Aouchiche and Hansen~\cite{AH} gave sharp lower and upper bounds for the  proximity and hence the minimum status of a graph as a function of its order,  and characterized the extremal graphs,  and
they also give a sharp lower bound for the  proximity and hence the minimum status of a graph with fixed order and diameter.  Lin et al.~\cite{LTSZ} found sharp lower and upper bounds for the minimum status of a graph with fixed order and maximum degree,  characterized extremal graphs attaining  the lower bound and found a necessary condition for graphs attaining the upper bound.
Alternate proofs were given by Rissner and Burkard~\cite{RB}.  Lin et al. ~\cite{LT2} extended the minimum status  to weighted graphs,  and they gave a formula for  the minimum status of the Cartesian product of two weighted graphs. Rissner and Burkard~\cite{RB} also showed that both radius and status obtain their minimum and maximum values at the same type of trees when order and maximum degree are fixed. Related work may be found in \cite{AH2, D,  D2,  Se}.

A matching of a graph $G$ is a set of edges in which any two distinct edges are not adjacent.  The matching number of $G$,  denoted by $m(G)$,  is the maximum of the cardinality of its matchings.  For a graph $G$ of order $n$,  $m(G)\leq \lfloor\frac{n}{2}\rfloor$.

A dominating set of a graph $G$ is a set $S$ of vertices such that each vertex of $G$ outside $S$ is adjacent to at least one vertex in $S$.  The domination number of $G$,  denoted by $\gamma (G)$,  is the  minimum cardinality of dominating sets of $G$.  If $G$ is a graph on $n$ vertices without isolated vertices,  then $\gamma(G)\leq \lfloor \frac{n}{2}\rfloor$,  see \cite{Ore}.  A dominating set of $G$ of cardinality $\gamma(G)$ is called $\gamma(G)$-set.  Bollob\' as and Cockayne \cite{B} showed that a graph $G$  without isolated vertices has a  $\gamma(G)$-set $S$ such that for each vertex $u \in S$,  there is a vertex of $G$ outside $S$  that is adjacent to $u$ but no other vertices in $S$,  and thus $m(G)\geq \gamma(G)$.

Following the above work, in this paper,  we find sharp lower and upper bounds for the minimum status of trees in term of  order and  matching number (domination number,  respectively),
and we also characterize all extremal cases. The upper bounds may be extended to graphs.

\section{Transformations}

In this section,  we give two types of transformations such that the minimum status is decreased or increased.

\begin{lemma}  \cite{KA,  LTSZ} \label{L22}
Let $T$ be a tree on  $n$ vertices and $x$ a vertex of $T$.  Then $x$ is in the median of $T$ if and only if $w_T(x)\leq \frac {n}{2}$.
\end{lemma}

For a graph $G$ with $uv\in E(G)$,  $G-uv$ denotes the graph obtained from $G$ by deleting the edge $uv$.  If $zw$ is an edge of the complement of a graph $G$,  then $G+zw$ denotes the graph obtained from $G$ by adding the edge $zw$. For $v\in V(G)$, by $N_G(v)$ we denote the set of vertices of $G$ that are adjacent to $v$, and by $d_G(v)$ we denote the degree of $v$ in $G$, i.e., $d_G(v)=|N_G(v)|$.

A pendant vertex is a vertex of degree one. The vertex adjacent to a pendant vertex is said to be a quasi-pendant vertex. A pendant edge is an edge incident to a pendant vertex. A cut edge is an edge whose removal disconnect the graph.
A non-pendant cut edge in a graph is a  cut edge that is not a pendant edge. For a graph $G$ with $u\in V(G)$ and $v\notin  V(G)$, if $G'$ is the graph with $V(G')=V(G)\cup \{v\}$ and $E(G')=E(G)\cup\{uv\}$, then we say that $G'$ is obtained from $G$ by attaching a pendant vertex at $u$.

\begin{proposition} \label{L23}
Let $G$ be a connected graph and $uv$ be a non-pendant cut edge of $G$.  Let $G_{uv}$  be the graph obtained from $G$ by contracting $uv$ to a vertex $u$ and attaching a  pendant vertex $v$ to $u$.  Then $s(G) >s(G_{uv})$.
\end{proposition}

\begin{proof}
Let $x$ be a vertex in the  median  of $G$.  Then $s(G)=s_G(x)$ and $s(G_{uv})\le  s_{G_{uv}}(x)$.   Let $G_1$ be  the component of $G-uv$ contains $u$ and $G_2$ be the component of $G-uv$ contains $v$.   Suppose without loss of generality that  $x\in V(G_1)$.  As we pass from $G$ to $G_{uv}$,  the distance between $x$ and a vertex of $V (G_1)\cup \{v\}$ remains unchanged and the distance between $x$ and a vertex of $V (G_2) \setminus \{v\}$ is decreased by $1$.  It follows that
\begin{eqnarray*}
s(G)-s(G_{uv}) &\geq & s_{G}(x)-s_{G_{uv}}(x)\\
&=& \sum_{w\in V(G)}\left(d_G(x,w)-d_{G_{uv}}(x,w)\right)\\
&=& \sum_{w\in V (G_2) \setminus \{v\}} 1\\
&=& |V(G_2)|-1\\
&>&0,
\end{eqnarray*}
i.e.,   $s(G) >s(G_{uv})$.
\end{proof}

A pendant path at $v$ in a graph $G$ is a path connecting $v$ and some vertex, say $w$ in $G$, such that $d_G(w)=1$, the degree of each internal vertex (if any exists) is two, and $d_G(v)\ge 2$.

\begin{proposition} \label{L26}
Let $T$ be a tree of order $n$ with $u\in T$ and $N_T(u)=\{u_1, \dots, u_k\}$,  where $k\geq 3$.  Let $B_i$ be the branch of $T$ at $u$ containing $u_i$ for $1\leq i\leq k$.   For $w\in V(B_2)$,  let $T'=T-\{uu_i: 3\leq i\leq  t\}+\{wu_i: 3\leq i\leq  t\}$,  where $3\leq t \leq k$.  Suppose that $|V(B_1)|\geq |V(B_2)|$.   Then $s(T')>s(T)$.
\end{proposition}

\begin{proof}
Since $|V(B_1)|+|V(B_2)|+2\leq n$ and $|V(B_1)|\ge |V(B_2)|$,  we have $|V(B_2)|\leq \frac{n}{2}-1$.  Let $x$ be a vertex in the median of $T$.   For any $v\in V(B_2)$,  we have  $w_T(v)\geq n-|B_2|\geq \frac{n}{2}+1$,  and thus  $x\notin V(B_2)$ by Lemma~\ref{L22}.

\noindent {\bf Case 1. } $x=u$.

In this case,  $w_{T'}(x) = \max\left\{w_T (x), \sum_{i=2}^{t}|V(B_i)|\right\}$.
Suppose first that  $\sum_{i=2}^{t}|V(B_i)| \leq \frac {n}{2}$. As $w_{T'}(u) \leq \frac {n}{2}$,  $u$ is in the median of $T'$ by Lemma~\ref{L22}.   As we pass from $T$ to $T'$,  the distance between $u$ and a vertex of $V (B_3) \cup \dots \cup V (B_t)$ is increased by $d_T(u, w)$,  and the distance between $u$ and any other vertex remains unchanged.  Thus
\[
s(T')-s(T)=s_{T'}(u)-s_T(u)=d_T (u, w)\sum_{i=3}^{t}|V(B_i)|>0,
\]
implying $s(T')>s(T)$.
Thus, we may suppose that $\sum_{i=2}^{t}|V(B_i)| > \frac {n}{2}$.
Then  there must exist a vertex $v \in V(B_2)$ such that $w_{T'} (v) \leq \frac {n}{2}$.  By Lemma \ref{L22},  $v$ is in the median of $T'$.
Note that
\[
d_{T'}(v,z)-d_T(u,z)=d_T (u, v) \mbox{ for } z\in V (B_1) \cup(V (B_{t + 1}) \dots \cup V (B_k)) \cup \{u\},
\]
%
%
%
\[
d_{T'}(v,z)-d_T(u,z)=d_{T}(v,z)-d_T(u,z)\ge -d_T (u, v) \mbox{ for } z\in V (B_2),
\]
%
and
\[
d_{T'}(v,z)-d_T(u,z)=d_T (v, w)  \mbox{ for } z\in V (B_3)\cup \dots \cup V(B_t).
\]
Thus
\begin{eqnarray*}
s(T')-s(T)&=& s_{T'}(v)-s_T(u)\\
&=&\sum_{z\in V(T')}d_{T'}(v,z)-\sum_{z\in V(T)} d_{T}(u,z)\\
 &\geq  & d_T (u, v)\left(|V(B_1)|+\sum_{i=t+1}^{k}|V(B_i)|+1-|V(B_2)|\right)\\
          &&+d_T (v, w)\sum_{i=3}^{t}|V(B_i)|\\
          &\ge & d_T (v,w)\\
          &>& 0,
\end{eqnarray*}
implying $s(T')>s(T)$.

\noindent {\bf Case 2. } $x\in \cup_{i=1}^kV(B_i)\setminus V(B_2)$.

In this case $w_{T'} (x)=w_T (x)$.  By Lemma \ref{L22},  $x$ is also in the  median of $T'$.

Suppose first that $x\in V(B_i)$ for  $i=1$ or $t+1\leq i \leq k$.  As we pass from $T$ to $T'$,  the distance between $x$ and a vertex of $V (B_3) \cup \dots \cup V (B_t)$ is increased by $d_T (u, w)$,  and the distance between $x$ and any other vertex remains unchanged.  It follows that
\[
s(T')-s(T)=s_{T'}(x)-s_T(x)=d_T (u, w)\sum_{i=3}^{t}|V(B_i)|>0,
\]
and thus $s(T')>s(T)$.

Next, suppose that  $x\in V(B_i)$ for $3 \leq i \leq t$,  As we pass from $T$ to $T'$,  the distance between $x$ and a vertex of $V (B_1)\cup (V (B_{t+1}) \dots \cup V (B_k))\cup \{u\}$ is increased by $d_T (u, w)$,  the distance between $x$ and a vertex $z$ of $V (B_2)$ is decreased by $d_T(u,z)-d_T(w,z)$, which is less than or equal to
$d_T (u, w)$,  and the distance between $x$ and any other vertex remains unchanged.  Thus
\begin{eqnarray*}
s(T')-s(T)&=&s_{T'}(x)-s_T(x)\\
& \geq & d_T (u, w)\left(|V(B_1)|+\sum_{i=t+1}^{k}|V(B_i)|+1-|V(B_2)|\right)\\
&\ge & d_T (u, w)\left(\sum_{i=t+1}^{k}|V(B_i)|+1\right)\\
&>&0,
\end{eqnarray*}
and thus $s(T')>s(T)$.

The result follows by combining the above cases.
\end{proof}

\section{Minimum status and matching number}

In this section,  we find sharp lower and upper bounds for the minimum status of a tree with fixed order and matching number and characterize the trees attaining these bounds. We note the upper bound may be extended for connected graphs.

By  $A_{n,m}$ we denote the tree  obtained from the star  $S_{n-m+1}$ by attaching a pendant edge to each of certain $m-1$ non-central vertices of $S_{n-m+1}$.   The center of $A_{n,m}$   is the center of  the star  $S_{n-m+1}$.


\begin{theorem} \label{NT1}
Let $T$ be a tree of  order $n$ with  matching number $m$,  where $1 \leq m \leq \lfloor \frac{n}{2}\rfloor$.  Then
$s(T ) \geq n+m-2$ with equality if and only if $T \cong A_{n,m}$.  
\end{theorem}

\begin{proof}
We show the result by induction on $n$.

First, we show the result for $n=2m$ by induction on $m$.  It is obvious that $T\cong P_{2m}\cong A_{2m, m}$ for $m=1,2$ with $s(T)=3m-2$. Suppose that $m\ge 3$ and the result follows for trees of order $2(m-1)$ with matching number $m-1$. Obviously, the diameter of $T$ is at least $4$. Choose a diametrical path  $v_0v_1v_2\dots v_d$ of $T$. As $n=2m$ and  the matching number of $T$ is $m$, we have
$d_T(v_0)=1$ and $d_T(v_1)=2$. By induction hypothesis, $s(T-v_0-v_1)\ge 3(m-1)-2=3m-5$ with equality if and only if $T-v_0-v_1\cong A_{2(m-1), m-1}$.
Let $x$ be a vertex of $T$ in its median. By Lemma~\ref{L22}, $w_T(x)\le m$, and $x$ is in the median of $T-v_0-v_1$ if $w_T(x)\le m-1$ or if $w_T(x)=m$ and $v_0$
belongs to the branch at $x$ with $m$ vertices. Suppose that $w_T(x)=m$ and $v_0$
does not belong to the branch $B$ at $x$ with $m$ vertices. Let $x^*$ be the neighbor of $x$ in $B$. Then $w_T(x^*)=m$ and $v_0$ belongs to the branch at $x^*$ with $m$ vertices.
By Lemma~\ref{L22},  $x^*$ is in the median of $T$ and  $T-v_0-v_1$. Thus, we may assume that  $x$ is in the median of $T$ and  $T-v_0-v_1$.
It follows that
\[
s(T)-s(T-v_0-v_1)=s_T(x)-s_{T-v_0-v_1}(x)=2d_T(x,v_2)+3\ge 3
\]
with equality if and only if $x=v_2$.
Thus
\[
s(T)\ge s(T-v_0-v_1)+3\ge 3m-5+3=3m-2
\]
with equalities if and only if $T-v_0-v_1\cong A_{2(m-1), m-1}$ and $x=v_2$, i.e., $T\cong A_{2m,m}$. This proves the result for $n=2m$.

Now, suppose that $n>2m$ and the result follows for trees of order $n-1$ with matching number $m$. Let $T$ be a tree of  order $n$ with  matching number $m$. As $n>2m$, we have by \cite[Lemma 2.7]{BG} that  there is an matching $M$ with $|M|=m$  and a pendant vertex $z$ such that $M$ does not meet $z$.
Let $yz\in E(T)$. It is evident that $M$ is a matching of $T-z$, and thus the matching number of $T-z$ is $m$. By induction hypothesis,  $s(G-z)\ge n-1+m-2=n+m-3$ with equality if and only if
$T-z\cong A_{n-1,m}$.
Let $x'$ be in the median of  $T$. Then by Lemma~\ref{L22}, $w_{T}(x')\le \frac{n}{2}$, and $x'$ is in the median of $T-z$ if $w_{T}(x')\le \frac{n-1}{2}$ or if $w_{T}(x')= \frac{n}{2}$ and $z$ belongs to the branch of $T$ with $w_T(x')=\frac{n}{2}$ vertices. Suppose that
$w_{T}(x')= \frac{n}{2}$ and $z$ does not belongs to the branch $B'$ of $T$ with $w_T(x')=\frac{n}{2}$ vertices. Then, as above,  the neighbor of $x'$ in $B'$ is in the media of $T$ and $T-z$. So we may assume that $x'$ is in the median of $T$ and $T'$.
Therefore
\[
s(T)-s(T-z)=s_T(x')-s_{T-z}(x')=d_T(x',z)\ge 1
\]
with equality if and only if $x'=y$. Thus
\[
s(T)\ge s(T-z)+1\ge n+m-2
\]
with equalities if and only if $T-z\cong A_{n-1,m}$ and $x'=y$, i.e., $T\cong A_{n,m}$.
\end{proof}

For integers $n$, $p$ and $q$ with $p\ge q\ge 0$ and $p+q+2\le n$, we define a graph called
 a dumbbell, denoted by  $D_n(p,q)$, as the graph formed by attaching $p$ pendant edges to one terminal vertex and $q$ pendant edges to the other terminal vertex of a path $P_{n-p-q}$. If $p+q=n-1$, then we define $D(n, p, q)=S_n$. Obviously, $D_n(1,1)\cong P_n\cong D_n(1,0)$.

\begin{lemma} \label{OL32}
If $p+q+2\le n$ and $p\geq q\geq 2$,  then $s(D_n(p,  q))>s(D_n(p+1,  q-1))$.
\end{lemma}

\begin{proof}
Let  $T=D_n(p, q)$.    Let $P=x_1\dots x_{n-p-q}$ be the path of $T$ connecting the two quasi-pendant vertices.  Let  $u$ be a pendant vertex adjacent to $x_{n-p-q}$.  Let  $T^*=T-x_{n-p-q}u+x_1u$.  Obviously,    $T^*\cong D_n(p+1, q-1)$.  Let $x$ be a vertex in the median of $T$.  Obviously, $x$ lies on the path $P$.
If $x=x_1$,  then $x_1$ is also  in the median of $T^*$ as $w_{T^*}(x_1)<w_T(x_1)\leq \frac{n}{2}$.   As we pass from $T$ to $T^*$,  the distance between $x_1$ and  $u$ is decreased by $n-(p+q+1)$, 
and the distance between  $x_1$ and any other vertex remains unchanged.  Thus $s(T)-s(T^*)=n-(p+q+1)>0$, implying $s(T)>s(T^*)$.  If $x=x_i$ with $i\ge 2$,  then $x_{i-1}$ is a median of $T^*$ as $w_{T^*}(x_{i-1})\leq w_T(x_i)\leq\frac{n}{2}$.   Let $B_1$ ($B_2$,  respectively)  be the component of $T-x_{i-1}x_i$ contains $x_i$ ($x_{i-1}$,  respectively).
Obviously, $|V(B_1)|=n-p-i+1$ and $|V(B_2)|=p+i-1$. Note that
\[
d_T(x_i,z)-d_{T^*}(x_{i-1},z)=-1 \mbox{ for } z\in V(B_1) \setminus \{u\},
\]
\[
d_T(x_i,z)-d_{T^*}(x_{i-1},z)=1 \mbox{ for } z\in V(B_2),
\]
%
and
\[
d_T(x_i, u)-d_{T^*} (x_{i-1}, u)=n-p-q+1-i-(i-1)=n-p-q+2-2i.
\]
Thus
\begin{equation*}
\begin{aligned}  s(T^*)-s(T)&=s_{T^*}(x_{i-1})-s_T(x_i)\\&=|V(B_1)|-1-|V(B_2)|-(n-p-q+2-2i)  \\  &=q-p-1\\ &<0,
\end{aligned}
\end{equation*}
i.e.,   $s(T)>s(T^*)$.
\end{proof}

 A quasi-pendant vertex is a vertex that is adjacent to a pendant vertex.  

\begin{theorem}\label{T32}
Let $T$ be a tree  on $n\geq 4$ vertices with  matching number $m$,  where $1 \leq m \leq \lfloor \frac{n}{2}\rfloor$.  Then \[ s(T )\leq m(n-m) \] with equality if and only if $T \cong  D_n(\lceil \frac{n+1}{2} \rceil-m, \lfloor \frac{n+1}{2} \rfloor-m)$.
\end{theorem}

\begin{proof}
Note that  the diameter $D_n\left(\left\lceil \frac{n+1}{2} \right\rceil-m, \left\lfloor \frac{n+1}{2} \right\rfloor-m\right)$ is $2m$  and  by Lemma~\ref{L22}, its center (the vertex of distance $m$ from a pendant vertex) is in its median. Thus,  we have
\begin{eqnarray*}
s\left(D_n\left(\left\lceil \frac{n+1}{2} \right\rceil-m, \left\lfloor \frac{n+1}{2} \right\rfloor-m\right)\right)
&=& 2\sum_{i=1}^{m-1} i+(n-2m+1)m\\
&=& m(n-m).
\end{eqnarray*}

Let $T$ be a tree  with order $n$ and  matching number $m$ such that its minimum status is as large as possible. From the value of the minimum status of $D_n(\lceil \frac{n+1}{2} \rceil-m, \lfloor \frac{n+1}{2} \rfloor-m)$,
 it suffices to show that $T \cong  D_n(\lceil \frac{n+1}{2} \rceil-m, \lfloor \frac{n+1}{2} \rfloor-m)$.

It is trivial for $m=1$.  Suppose that $m \geq 2$.   Let $\alpha$ be the number of  quasi-pendant vertices in $T$. As $m\ge 2$, the diameter of $T$ is at least $3$, and thus $\alpha\ge 2$.

We claim that $\alpha=2$.  Otherwise,  $\alpha\ge 3$,  there are at least three branches, say $B_v$, $B_w$ and $B_z$, at some vertex $u$,  and at least two of them  are nontrivial, where $v,w,z\in N_T(u)$.  Assume that $|V(B_v)|\ge |V(B_w)|\ge |V(B_z)|$.
Obviously, $B_v$ and $B_w$ are non-trivial.
Let $M$ be a matching of $T$ with $|M|=m$. 
Suppose that $uz\notin M$.  Let $u_1$ be a quasi-pendant vertex of $B_w$.  Then  $T'=T-uz+u_1z$ is a tree on $n$ vertices with matching number $m$.  By Proposition~\ref{L26},  $s(T')>s(T)$,  a contradiction.
Thus,  $uz\in M$, implying $uv, uw\notin M$.
If $B_z$ is nontrivial, then by reversing the roles of $w$ and $z$ as above,  we have a contradiction.
So  $B_z$ is trivial.  Let  $T''=T-uv+zv$. Then $m(T'')=m$,  and  by Proposition~\ref{L23} or \ref{L26},  we have $s(T'')>s(T)$,  a contradiction.
Therefore, we have $\alpha=2$, as claimed.
It follows that $T \cong D_n(p,  q)$  for some $p$ and $q$ with $p\ge q\ge 1$.  If $p=1$,  then $T\cong P_n$, $m=\lfloor \frac{n}{2}\rfloor$, and thus $T\cong  D_n(1,1)\cong D_n(1,0)$.

Suppose $p>1$. Assume that $T=D_n(p,q)$. Let $P_d=v_1\dots v_{d}$ be a diametrical path of $T$ with $d_T(v_2)=p+1$, where $d=n-p-q+2$. Note that $m=\lfloor\frac{d}{2}\rfloor$,  we have $d=2m, 2m+1$. If $d=2m$, then the matching number of $D_n(p-1,q)$ is $m$, and
 by Proposition~\ref{L23}, we have $s(D_n(p-1,q))>s(D_n(p,q))=s(T)$, a contradiction. Thus, $d=2m+1$. Then $2m+1=n-p-q+2$,  i.e., $p+q=n+1-2m$.
  By  Lemma \ref{OL32},  we have $T \cong  D_n(\lceil \frac{n+1}{2} \rceil-m, \lfloor \frac{n+1}{2} \rfloor-m)$.
\end{proof}

Let $G$ be a connected graph on $n$ vertices with matching number $m$,  where $1\leq m \leq \lfloor \frac{n}{2}\rfloor$.  Obviously, $s(G)\le s(T)$ for a spanning tree $T$ of $G$. Thus, by Theorem~\ref{T32},
 \[ s(G )\leq m(n-m)  \] with equality if $G \cong  D_n(\lceil \frac{n+1}{2} \rceil-m, \lfloor \frac{n+1}{2} \rfloor-m)$.

%

\section{Minimum status and domination number}

In this section,  we find sharp lower and upper bounds for the minimum status of a tree with fixed order and domination number, and characterize the trees attaining these bounds. The upper bounds may be extended for connected graphs.

\begin{theorem}
Let $T$ be a tree with order $n$ and domination number $\gamma$,  where $1 \leq \gamma\leq \lfloor \frac{n}{2}\rfloor$.  Then\[s(T) \geq n+\gamma-2 \] with equality if and only if $T\cong A_{n,\gamma}$.
\end{theorem}

\begin{proof}
Note that $m(T) \geq \gamma$ and $\gamma(A_{n,\gamma}) = \gamma$.  By Theorem \ref{NT1} and Proposition~\ref{L23},  we have $s(T) \geq s\left(A_{n, m(T)}\right) \geq s(A_{n,\gamma})=n+\gamma-2$ with equalities if and only if $T\cong A_{n, m(T)}$ and $m(T) =\gamma$,  i.e.,   $T \cong A_{n,\gamma}$.
\end{proof}

The proximity $\pi= \pi(G)$ of a connected graph $G = (V ,  E)$ is the minimum,  over all vertices,  of the average distance from a vertex to all others.   Obviously,  $s(G)=(n-1)\pi$.


In \cite{AH}, an upper bound for the proximity was given, which is restated as below:
Let $G$ be a graph of order $n\ge 3$.  Then $s(G)\leq \lfloor\frac{n^2}{4}\rfloor$
with equality if and only if G is either the cycle $C_n$ or the path $P_n $. Observe that  $\gamma(C_n)= \gamma(P_n)= \lceil \frac {n}{3}\rceil$.

In the rest of this section,  for trees with order $n$ and  domination number $\gamma$,   we consider  $1  \leq \gamma < \lceil \frac {n}{3}\rceil$  and $\lceil \frac {n}{3}\rceil <\gamma \leq \lfloor \frac {n}{2}\rfloor$,  separately.

\begin{theorem}\label{T42}
Let $T$ be a tree with order $n$ and domination number $\gamma$,  where $1  \leq \gamma < \lceil \frac {n}{3}\rceil$.  Then
\[s(T)\leq
\begin{cases}
\frac{3\gamma-1}{2}\left(n-\frac{3\gamma-1}{2}\right)   &\text{if $\gamma$ is odd }\\[2mm]
 \frac{3\gamma}{2}\left(n+1-\frac{3\gamma}{2}\right)-\left\lceil\frac{n}{2}\right\rceil &\text{if $\gamma$ is even}
 \end{cases}
 \]
with equality if and only if $T\cong D_n(\lceil \frac {n-3\gamma+2}{2}\rceil,  \lfloor \frac {n-3\gamma+2}{2}\rfloor)$.
\end{theorem}

\begin{proof} Let $H=D_n(\lceil \frac {n-3\gamma+2}{2}\rceil,  \lfloor \frac {n-3\gamma+2}{2}\rfloor)$. Obviously, the diameter of $H$ is $3\gamma-1$.  Let $uv$ be a pendant edge with $d_H(v)=\lceil \frac {n-3\gamma+2}{2}\rceil+1$. Let $x$ be a vertex of $H$ with $d_H(x,u)=\lfloor\frac{3\gamma-1}{2}\rfloor$. Obviously, $x$ is in the median of $H$ by Lemma~\ref{L22}.
If $\gamma$ is odd, then
\begin{eqnarray*}
s(H) &= & s_H(x)\\
&=&s(P_{3\gamma})+\frac{3\gamma-1}{2}\cdot (n-3\gamma)\\
&=&\frac{(3\gamma)^2-1}{4}+\frac{3\gamma-1}{2}\cdot (n-3\gamma)\\
&=&\frac{3\gamma-1}{2}\left(n-\frac{3\gamma-1}{2}\right).
\end{eqnarray*}
If $\gamma$ is even, then
\begin{eqnarray*}
s(H)&=& s(P_{3\gamma})+\left(\frac{3\gamma}{2}-1\right)\left\lceil \frac{n-3\gamma}{2} \right\rceil   +\frac{3\gamma}{2}\cdot \left\lfloor\frac{n-3\gamma}{2}\right\rfloor\\
&=&\frac{(3\gamma)^2}{4}+\frac{3\gamma}{2}\cdot (n-3\gamma)-\left\lceil \frac{n-3\gamma}{2} \right\rceil \\
    &=& \frac{3\gamma}{2}\left(n+1-\frac{3\gamma}{2}\right)-\left\lceil\frac{n}{2}\right\rceil.
\end{eqnarray*}
That is, we have
\begin{eqnarray*}
&&s\left(D_n\left(\left\lceil \frac {n-3\gamma+2}{2}\right\rceil,  \left\lfloor \frac {n-3\gamma+2}{2}\right\rfloor\right)\right)\\
&=&\begin{cases}
\frac{3\gamma-1}{2}\left(n-\frac{3\gamma-1}{2}\right)   &\text{if $\gamma$ is odd, }\\[2mm]
 \frac{3\gamma}{2}\left(n+1-\frac{3\gamma}{2}\right)-\left\lceil\frac{n}{2}\right\rceil &\text{if $\gamma$ is even.}
 \end{cases}
\end{eqnarray*}

Let $T$ be a tree with order $n$ and  domination number at most $\gamma$ such that its minimum status is as large as possible. By the value of $s\left(D_n(\lceil \frac {n-3\gamma+2}{2}\rceil,  \lfloor \frac {n-3\gamma+2}{2}\rfloor)\right)$, we only need to show that
$T\cong D_n(\lceil \frac {n-3\gamma+2}{2}\rceil,  \lfloor \frac {n-3\gamma+2}{2}\rfloor)$.

It is trivial if $\gamma=1$. Suppose that $\gamma\ge 2$.

\noindent
{\bf Claim 1.} $T$ has exactly two quasi-pendant vertices.

Otherwise,  there are at least three branches at some vertex $u$ of $T$, and  at least two of them, say $B_x$ and $B_y$, are nontrivial, where $x,  y \in N_T(u)$.
Suppose without loss of generality that $|V(B_x)| \geq |V(B_y)|$.
Let $S$ be a $\gamma(T)$-set.  Suppose  that $u \in S$.
Let $w \in V(B_y) \cap S$.  For a vertex $z \in N_T(u)$ with $z \neq x,  y$,  let $T'= T-uz+wz$. Obviously, $T'$ is a tree.  As $w \in S$, $S$ is a dominating set of $T'$, implying  $\gamma(T')\leq |S|=\gamma(T) \leq \gamma$.   By Proposition~\ref{L26}, $s(T'') > s(T)$,  a contradiction.
It follows that  $u \notin S$. Then there is a vertex $x'\in N_T (u) \cap S$.  Let $\{y',  z'\}\subseteq N_T (u) \setminus \{x'\}$.  Let $T''= T-uz' +y'z'$  if $|V(B_{x'})| \geq |V(B_{y'})|$,  and $T''= T-uz'+ x'z'$ if $|V(B_{x'})| < |V(B_{y'})|$.  Then $T''$ is a tree with $\gamma(T')\leq \gamma(T)\leq \gamma$.  By Proposition~\ref{L26},  $s(T'') > s(T)$, also a contradiction.
Therefore,  $T$ has exactly two quasi-pendant vertices, as claimed.

By Claim 1,  $T\cong D_n(p,q)$ for some $p$ and $q$ with $p\geq q\geq 1$.  As $\gamma(T)\le \gamma<\lceil \frac {n}{3}\rceil$,  we have $p \geq 2$.

\noindent
{\bf Claim 2.} $\gamma(T)=\gamma$.

Otherwise, we have  $\gamma(T) < \gamma$.  Let $T^*= D_n(p-1,  q)$ if $p-1\geq q$ and $T^*=D_n(p,  q-1)$ if $p=q$. Evidently,  $\gamma(T^*) \leq \gamma(T) + 1 \leq \gamma$. By Proposition~\ref{L23},  we have  $s(T) < s(T^*)$,  a contradiction.  Hence  $\gamma(T)=\gamma$, as claimed.

As $T\cong D_n(p,q)$ and  $\gamma(T)=\gamma$ by Claim 2, we have  $\lceil\frac{n-p-q+2}{3}\rceil=\gamma$, i.e., $n-p-q = 3\gamma -2$,  $3\gamma -3$,   $3\gamma -4$.

\noindent
{\bf Claim 3.} $n-p-q=3\gamma-2$.

Otherwise,
we have  $n-p-q= 3\gamma -3, 3\gamma -4$.  Let $\widehat{T} = D_n(p-1,  q)$  if $p-1\geq q$ and $\widehat{T}= D_n(p,  q-1)$ for $p=q$.   Obviously, $\gamma(\widehat{T})=\lceil\frac{n-p-q+3}{3}\rceil =\gamma$.   By Proposition~\ref{L23}, $s(T) < s(\widehat{T})$,  a contradiction.

By Claim 3, we have $p+q = n-3\gamma + 2$. Now by Lemma \ref{OL32},  we have $T\cong D_n\left(\lceil \frac {n-3\gamma+2}{2}\rceil,  \lfloor \frac {n-3\gamma+2}{2}\rfloor\right)$.
\end{proof}

Let $G$ be a connected graph with order $n$ and domination number $\gamma$ ,  where $1  \leq \gamma < \lceil \frac {n}{3}\rceil$.  Then
\[s(G)\leq
\begin{cases}
\frac{3\gamma-1}{2}\left(n-\frac{3\gamma-1}{2}\right)   &\text{if $\gamma$ is odd }\\[2mm]
 \frac{3\gamma}{2}\left(n+1-\frac{3\gamma}{2}\right)-\left\lceil\frac{n}{2}\right\rceil &\text{if $\gamma$ is even}
 \end{cases}
 \]
with equality if  $G\cong D_n(\lceil \frac {n-3\gamma+2}{2}\rceil,  \lfloor \frac {n-3\gamma+2}{2}\rfloor)$.


A caterpillar is a tree in which removal of all pendant vertices gives a path.
Let $C_n(p,q)$ be the  caterpillar  obtained by attaching a pendant vertex $v_i'$ to $v_i$ of the path $P_{n-p-q} =v_1v_2\dots v_{n-p-q}$ for $i = 1,   \dots,  p,  n-p-2q+1,  \dots ,  n-p-q$,  where $p\geq q \geq 1$ and $2(p+q) \leq n$.

\begin{lemma} \label{OL41} If $p\geq q+2$ and $2(p + q) < n$,  then $s(C_n( p-1,  q+1)) > s(C_n(p,q))$.
\end{lemma}

\begin{proof}
Let $T=C_n(p- 1,  q+1)$, as labelled above. Let  $T'=T-v_{n-p-2q}v_{n-p-2q}'+v_pv_{n-p-2q}'$. Obviously,  $T'\cong C_n(p,  q)$.

\noindent
{\bf Case 1.} $2(p-1)<\lceil\frac{n}{2}\rceil$.

Let  $x=v_{\lceil \frac{n}{2}\rceil-p+1}$, then by Lemma~\ref{L22}, $x$  is in the median of $T$ as $w_T(x)\leq \frac{n}{2}$.  As we pass $T$ to $T'$,  the distance between $x$ and $v_{n-p-2q}'$ is decreased by  $[n-p-2q-(\lceil \frac{n}{2}\rceil-p+1)+1]-(\lceil \frac{n}{2}\rceil-p+1-p+1)=2(p-q-1)+(n-2\lceil \frac{n}{2}\rceil)$,  and the distance between the $x$ and any other vertex  remains unchanged.   Thus $s(T)-s(T')\ge s_T(x)-s_{T'}(x)=2(p-q-1)+(n-2\lceil \frac{n}{2}\rceil)>0$, implying $s(T)>s(T')$.

\noindent
{\bf Case 2.} $2(p-1)\ge\lceil\frac{n}{2}\rceil$.

Let  $x=v_{\lceil \frac{n}{4}\rceil}$, then by Lemma~\ref{L22}, $x$  is in the median of $T$ as $w_T(x)\leq \frac{n}{2}$.  As we pass $T$ to $T'$,  the distance between $x$ and $v_{n-p-2q}'$ is decreased  by  $n-p-2q+1-\lceil \frac{n}{4}\rceil -\left(p+1-\lceil \frac{n}{4}\rceil\right)=n-2(p+q)$,  and the distance between  $x$ and any other vertex  remains unchanged.   Thus $s(T)-s(T')\ge s_T(x)-s_{T'}(x)=n-2(p+q)>0$,  implying $s(T)>s(T')$.
\end{proof}

\begin{lemma} \label{OL42}
Let $T$ be a caterpillar on $n$ vertices with $r$ pendant vertices,  where $2 \leq r \leq  \lfloor \frac {n}{2}\rfloor$ and each vertex of  $T$ is adjacent to at most one pendant vertex.  Then  $s(T) \leq s(C_n(\lceil \frac {r}{2}\rceil, \lfloor \frac {r}{2}\rfloor ))$ with equality if and only if   $T \cong C_n(\lceil \frac {r}{2}\rceil, \lfloor \frac {r}{2}\rfloor )$.
\end{lemma}

\begin{proof}
It is trivial for $r=2,  \frac {n}{2}$ as then   $T\cong C_n(\lceil \frac {r}{2}\rceil, \lfloor \frac {r}{2}\rfloor )$.    Suppose that $2<r<\frac {n}{2}$.  Obviously,  the diameter of $T$ is $n-r+1$.
Let $T$ be a caterpillar satisfying the conditions of the lemma such that its minimum  status is as large as possible.  Let $v_1\dots v_{n-r+2}$ be a diametrical path of $T$ and let $U=\{v\in V(T): d_T(v)=2\} \setminus \{v_2, v_{n-r+1}\}$.
We claim that $T-U$ has exactly two nontrivial components.  Otherwise,  there are three vertices $v_i$,  $v_j$,  and $v_k$ in $T$ such that $d(v_k)=3$ and $\{v_i,  v_j\}\subseteq U$,  where $i < k < j$.  Let $v_k'$ be the unique pendant vertex adjacent to  $v_k$.  Let $B_1$,  $B_2$ be two nontrivial branches of $T$ at $v_k$ with $v_i \in V(B_1)$ and $v_j \in V(B_2)$.  Suppose without loss of generality that $|V(B_1)|\geq |V(B_2)|$.  Let $T'=T-v_kv_k'+v_jv_k'$. Evidently, $T'$ is a caterpillar on $n$ vertices with $r$ pendant vertices, and each vertex of  $T'$ is adjacent to at most one pendant vertex.
By Proposition~\ref{L26},  $s(T') > s(T)$, which is  a contradiction.  Thus $T-U$ has exactly two  nontrivial components, as claimed. That is,   $T\cong C_n(p, q)$ for some $p,q$ with $p+ q= r$.  By Lemma \ref{OL41},  $T\cong C_n(\lceil \frac {r}{2}\rceil, \lfloor \frac {r}{2}\rfloor)$.
\end{proof}

\begin{lemma} \cite{D2}  \label{OL43}
Let $T$ be a tree on $n$ vertices with domination $\gamma$,  where $\gamma> \lfloor\frac{n}{3}\rfloor$.  Then the diameter of $T$ is at most $2n-3\gamma+1$.
\end{lemma}

\begin{theorem}
Let $T$ be a tree on $n$ vertices with domination number $\gamma$,  where $\lceil \frac {n}{3}\rceil < \gamma \leq \lfloor \frac {n}{2}\rfloor$.  Then
\[
s(T) \leq3n\gamma +3\gamma-n-\left\lceil\frac{n^2+18\gamma^2}{4}\right\rceil
\]
with equality if and only if $T\cong C_n\left(\lceil \frac {3\gamma-n}{2}\rceil, \lfloor \frac {3\gamma-n}{2}\rfloor\right)$.
\end{theorem}

\begin{proof}
Let $a=\lceil \frac {3\gamma-n}{2}\rceil$, $b=\lfloor \frac {3\gamma-n}{2}\rfloor$ and $c=\lceil\frac{2n-3\gamma}{2}\rceil$. Obviously, $a+b=3\gamma-n$, and $a=b,b+1$.
Let $H=C_n(a, b)$, as labelled as before. Let $x=v_c$. As  $w_H(x)\leq \frac{n}{2}$,  $x$  is in the median of $H$ by Lemma~\ref{L22}.  Let
 $A=\{v_i:i=1,\dots,{2n-3\gamma}\}$ and $B=\{v_i':i=1,\dots,  a, {2n-3\gamma-b+1}, \dots ,  {2n-3\gamma}\}$.
By direct calculation, we have
\[
\sum_{u\in A}d_H(x,u)=s(P_{2n-3\gamma})=\left\lfloor \frac{(2n-3\gamma)^2}{4}\right\rfloor
\]
and
\begin{eqnarray*}
\sum_{u\in B}d_H(x,u)
&=& \sum_{i=1}^a(c-a+i)+\sum_{i=1}^b(2n-3\gamma-b+1-c+i)\\
&=& \begin{cases}
a(3n-6\gamma+1)+a(a+1)  & \mbox{if } a=b\\
a(3n-6\gamma+1)+a^2-(2n-3\gamma-b+1-c) &  \mbox{if } a=b+1
\end{cases}\\
 &=&\begin{cases}
 \frac{(3\gamma-n)(5n-9\gamma+4)}{4}   &\text{if $3\gamma-n$ is even,}\\[2mm]
 \frac{(3\gamma-n+1)(5n-9\gamma+3)}{4}-\left(\lfloor \frac{2n-3\gamma}{2}\rfloor+1-\frac{3\gamma-n-1}{2}\right)   &\text{if $3\gamma-n$ is odd.}
 \end{cases}
 \end{eqnarray*}
Thus
\begin{eqnarray*}
s(H)&=&\sum_{u\in A}d_H(x,u)+\sum_{u\in B}d_H(x,u)\\
&=& \left\lfloor \frac{(2n-3\gamma)^2}{4}\right\rfloor\\
&&+\begin{cases}
 \frac{(3\gamma-n)(5n-9\gamma+4)}{4}   &\text{if $3\gamma-n$ is even}\\[2mm]
 \frac{(3\gamma-n+1)(5n-9\gamma+3)}{4}-\left(\lfloor \frac{2n-3\gamma}{2}\rfloor+1-\frac{3\gamma-n-1}{2}\right)   &\text{if $3\gamma-n$ is odd}
 \end{cases}\\[2mm]
&=&3n\gamma +3\gamma-n-\left\lceil\frac{n^2+18\gamma^2}{4}\right\rceil.
\end{eqnarray*}
So  \[s\left(C_n\left(\left\lceil \frac {3\gamma-n}{2}\right\rceil,\left\lfloor \frac {3\gamma-n}{2}\right\rfloor\right)\right)=3n\gamma +3\gamma-n-\left\lceil\frac{n^2+18\gamma^2}{4}\right\rceil.\]

Let $T$ be a tree with order $n$ and  domination number at least $\gamma$ such that its minimum status is as large as possible. By the value of $s\left(C_n\left(n,\left\lceil \frac {3\gamma-n}{2}\right\rceil,\left\lfloor \frac {3\gamma-n}{2}\right\rfloor\right)\right)$, it suffices to show that $T\cong C_n\left(\lceil \frac {3\gamma-n}{2}\rceil,\lfloor \frac {3\gamma-n}{2}\rfloor\right)$.

\noindent {\bf Claim 1.}  Each vertex of $T$ is adjacent to  at most one pendant vertex.

Otherwise,  there is a vertex $u$ adjacent to  two pendant vertices,  say $v$ and $w$.  Let $T'=T-uv + vw$.  For any $\gamma(T')$-set $S$,  it contains one of $v$ or $w$,  and thus $S \cup \{u\} \setminus \{x\}$ with $x = v, w$ is a dominating set of $T$, implying
 $\gamma\le \gamma(T)\le |S \cup \{u\} \setminus \{x\}|\le |S|=\gamma(T')$.  By Proposition~\ref{L23},  $s(T') > s(T)$,  a contradiction. Thus, each vertex of $T$ is adjacent to  at most one pendant vertex, as claimed.


\noindent {\bf Claim 2.}   If $d_T(u) \geq 3$,  then $\gamma(T-uz)=\gamma(T)$ for $uz \in E(T)$.

Otherwise,  $\gamma(T-uz)=\gamma(T) + 1$ as $\gamma(T-e)-\gamma(T) = 0, 1$ for $e \in E(T)$.  Let $B_v$ be the branch of $T$ at $u$ containing $v$ for $v \in N_T(u)$.  Let $\{x,  y\}\subseteq N_T (u) \setminus \{z\}$.  Suppose without loss of generality that $|V(B_x)| \geq |V(B_y)|$.  Then $T'= T-uz + yz$ is a tree with $\gamma(T')\geq \gamma(T'-yz)-1= \gamma(T-uz)-1 = (\gamma(T)+1)-1 =\gamma(T) \geq \gamma$. By Proposition~\ref{L26},  $s(T') > s(T)$,  a contradiction. It follows that  $\gamma(T-uz)=\gamma(T)$ for $uz \in E(T)$ if $d_T(u)\ge 3$, as claimed.

\noindent {\bf Claim 3.}    $T$ is a caterpillar.

Otherwise,  there are at least three nontrivial branches $B_x$,  $B_y$,  $B_z$ of $T$ at $u$ containing $x$,  $y$,  $z$, respectively,  where $\{x,  y,  z\} \subseteq N_T (u)$.  Since $d_T(u)\geq 3$, we have by Claim 2 that $\gamma (T-ux) = \gamma(T-uy) =\gamma(T-uz) =\gamma(T)$.  Note that $\gamma(T-w)-\gamma(T) \geq-1$ for $w\in V(T)$.  Suppose that $\gamma(B_x-x)-\gamma(B_x) =\gamma(B_y-y)-\gamma(B_y) = \gamma (B_z-z)-\gamma (B_z) =-1$.  Then, for each  $i =x, y, z$, there is a $\gamma(B_i)$-set $S_i$ such that $i \in S_i$ and $S_i \setminus \{i\}$ is a $\gamma(B_i-i)$-set.  Let $T^*$ be the tree obtained from $T$ by deleting the vertices in $V(B_x)$ and let $S$ be a $\gamma (T^*)$-set.  If $u\notin S$,  then $S'=(S\setminus V(B_y)\cup V(B_z))\cup(S_y\setminus\{y\})\cup(S_z\setminus\{z\})\cup\{u\}$ is a dominating set of $T^*$ with cardinality $|S'| = |S|-\gamma (B_y)-\gamma(B_z)+ \gamma(B_y-y)+ \gamma(B_z-z)+1=|S|-1=\gamma(T^*)-1 <\gamma(T^*)$,  which is impossible. It follows that  $u \in S$. Then  $S \cup(S_x \setminus \{x\})$ is a dominating set of $T$ with cardinality $\gamma (T^*) +\gamma(B_x-x)$.  Therefore $\gamma(T)\leq \gamma (T^*) +\gamma (B_x -x)$.  As $\gamma(T-ux) = \gamma(T^*) + \gamma(B_x)$, we have  $(T-ux)-\gamma(T) \geq \gamma(B_x)-\gamma(B_x-x)=1$,  which implies that $\gamma(T-ux)-\gamma (T) = 1$,  a contradiction to the fact that $\gamma(T-ux)=\gamma (T)$.  Therefore there is a vertex in $\{x,  y,  z\}$,  say $z$,  such that $\gamma(B_z-z)-\gamma(B_z) >-1$.  Assume that $|V(B_x)| \geq |V(B_y)|$.  Let $w \in V(B_y)$ be a quasi-pendant vertex of $T$.  Let $T'= T-uz + wz$.  As $w$ is adjacent to  a unique pendant vertex, there is a $\gamma(T')$-set $R$ containing $w$.  If $z \in R$,  then $R$ is also a dominating set of $T'-wz$,  and so $\gamma(T'-wz)\leq \gamma(T')$.  If $z \notin  R$,  then $R \cap V(B_z-z)$ is a $\gamma(B_z-z)$-set,  and so for a $\gamma(B_z)$-set $R_1$,  $(R \setminus R \cap V(B_z-z)) \cup  R_1$ is a dominating set of $T'-wz$ with cardinality $|R|-\gamma(B_z-z) + \gamma(B_z)\leq |R|$, implying that $\gamma(T'-wz)\leq \gamma(T')$.  Therefore  $\gamma(T') = \gamma(T'-wz) = \gamma(T-uz) = \gamma(T) \geq \gamma$.  By Proposition~\ref{L26},  $s(T') > s(T)$,  a contradiction. This proves Claim 3.

Let $r$ be the number of pendant vertices of $T$.  By Claims 1 and 3 and Lemma \ref{OL42},   we have $T\cong C_n(\lceil \frac {r}{2}\rceil, \lfloor \frac {r}{2}\rfloor)$.
Let $d$ the diameter of $T$.  Then $r = n-d + 1$.  As $\gamma(T) > \lceil \frac {n}{3}\rceil$,  we have by Lemma \ref{OL43} that $d \leq 2n-3\gamma(T) + 1$.  Thus $r\geq 3\gamma(T) -n \geq 3\gamma-n$. Note that $C_n(\lceil \frac {3\gamma-n}{2}\rceil, \lfloor \frac {3\gamma-n}{2}\rfloor)$ has exactly $3\gamma-n$ pendant vertices  and its domination number is $\gamma$.
By Proposition~\ref{L23},  if $r>3\gamma-n$, then $s(C_n(\lceil \frac {r}{2}\rceil, \lfloor \frac {r}{2}\rfloor ))<s(C_n(\lceil \frac {3\gamma-n}{2}\rceil, \lfloor \frac {3\gamma-n}{2}\rfloor ))$, a contradiction. Therefore $r=3\gamma-n$, i.e.,  $G\cong C_n\left(\lceil \frac {3\gamma-n}{2}\rceil, \lfloor \frac {3\gamma-n}{2}\rfloor \right)$.
\end{proof}

Let $G$ be a connected graph on $n$ vertices with domination $\gamma$,  where $\lceil \frac {n}{3}\rceil < \gamma \leq \lfloor \frac {n}{2}\rfloor$.  Then
\[
s(G) \leq3n\gamma +3\gamma-n-\left\lceil\frac{n^2+18\gamma^2}{4}\right\rceil
\]
with equality if  $G\cong C_n\left(\lceil \frac {3\gamma-n}{2}\rceil, \lfloor \frac {3\gamma-n}{2}\rfloor\right)$.

\section{Concluding remarks}

We present sharp lower and upper bounds on the minimum status of a tree using order and matching number (domination number respectively). The trees that attain these bounds are determined. The minimum status is a fundamental graph parameter  to measure the centrality  of a graph or network \cite{AH,BH,RB}.
The notion of centrality has been widely used in many different areas.
Some other parameters, like radius \cite{BH,RB}, average distance \cite{Do}, and distance spectral radius \cite{AH1},  also play roles in the measurement of centrality of a graph. For example, Rissner and Burkard \cite{RB} also established analogous results on trees with minimum and maximum radius. It is of interest to investigate the relationship among these parameters.


\vspace{5mm}

\noindent {\bf Acknowledgement.} This work was supported by the National Natural Science Foundation of China (No.~11671156 and No.~71801186).

\end{document}